\documentclass[12pt,technote, onecolumn, draft]{IEEEtran}
\usepackage[utf8]{inputenc}
\usepackage{latexsym}
\usepackage{mathrsfs}
\usepackage{graphicx}
\usepackage{multirow}
\usepackage{amsfonts,amssymb,amsmath,amsthm,bm}
\usepackage{color}
\usepackage{booktabs}
\usepackage{cite}
\newtheorem{theorem}{Theorem}[section] 
\newtheorem{definition}[theorem]{Definition} 
\newtheorem{lemma}[theorem]{Lemma} 
\newtheorem{corollary}[theorem]{Corollary}
\newtheorem{example}[theorem]{Example}
\newtheorem{proposition}[theorem]{Proposition}
\newtheorem{remark}[theorem]{Remark}

\begin{document}
	\title{A class of pseudorandom sequences From Function Fields $^{\dag}$}
	\author{Xiaofeng Liu, Jun Zhang, Fang-Wei Fu
		\IEEEcompsocitemizethanks{\IEEEcompsocthanksitem Xiaofeng Liu and Fang-Wei Fu are with Chern Institute of Mathematics and LPMC, Nankai University, Tianjin 300071, China, Emails: lxfhah@mail.nankai.edu.cn, fwfu@nankai.edu.cn. Jun Zhang are with School of Mathematical Sciences, Capital Normal University, Beijing 100048, China, Email: junz@cnu.edu.cn
		}
		\thanks{$^\dag$This research is supported by the National Key Research and Development Program of China (Grant Nos.  2022YFA1005000), the National Natural Science Foundation of China (Grant Nos. 12141108, 61971243, 12226336), the Natural Science Foundation of Tianjin (20JCZDJC00610), the Fundamental Research Funds for the Central Universities of China (Nankai University), and the Nankai Zhide Foundation.}
		\thanks{}
	}{\tiny }
		\maketitle

	\begin{abstract}
		   Motivated by the constructions of pseudorandom sequences over the cyclic elliptic function fields by Hu \textit{et al.} in \text{[IEEE Trans.  Inf. Theory, 53(7), 2007]} and the constructions of low-correlation, large linear span binary sequences from function fields by Xing \textit{et al.} in \text{[IEEE Trans. Inf. Theory, 49(6), 2003]}, we utilize the bound derived by Weil \text{[Basic Number Theory, Grund. der Math. Wiss., 
           Bd 144]} and Deligne \text{[ Lecture Notes in Mathematics, vol. 569 (Springer, Berlin, 1977)]} for the exponential sums over the general algebraic function fields and study the periods, linear complexities, linear complexity profiles, distributions of $r-$patterns, period correlation and nonlinear complexities for a class of $p-$ary sequences that generalize the constructions in \text{[IEEE Trans. Inf. Theory, 49(6), 2003]} and [IEEE Trans.  Inf. Theory, 53(7), 2007].
           \end{abstract}

	\begin{IEEEkeywords}
		 Exponential sums over the algebraic function fields, automorphism groups, Cyclic elliptic function fields, Riemann-Roch theorem, Rational function fields, Artin-Schreier $\mathbb{F}_{q}-$covering.
	\end{IEEEkeywords}
	
	\section{Introduction}
	\label{sec:1}
The construction of $p$-ary sequences that simultaneously exhibit low correlation and high linear as well as nonlinear complexity has been widely applied in code division multiple access (CDMA), digital communication systems, word-based stream ciphers, broadband satellite communications, \textit{etc.}
  
We now review several approaches to constructing pseudorandom sequences using the framework of function fields.  
A standard technique relies on the cyclic structure of finite fields. Let $\mathbb{F}_{q}$ be a finite field with $q$ elements and let $\alpha$ be a primitive element of $\mathbb{F}_{q}$, where $q=p^{e}$ with $e\ge 1$ and $p$ prime. For any polynomial $f(x)\in\mathbb{F}_{q}[x]$, define the sequence $S_{f}$ by
\[
    S_{f}=\{\mathrm{Tr}(f(\alpha^{t}))\}_{t=0}^{\infty},
\]
where $\mathrm{Tr}$ is the absolute trace from $\mathbb{F}_{q}$ to $\mathbb{F}_{p}$. Set
\[
    \tilde{S}=\{S_{f}\mid f(x)\in\mathbb{F}_{q}[x],\ \deg(f(x))\le d\}.
\]
Throughout this paper, for each positive integer $\ell$, we denote by $\omega_{\ell}$ the complex primitive $\ell$-th root of unity $e^{2\pi i/\ell}$. For any $S_{f},S_{g}\in\tilde{S}$, the periodic correlation at shift $\tau$ is defined as
\[
    C_{f,g}=\sum_{t=0}^{q-2}\omega_{p}^{\mathrm{Tr}(f(\alpha^{t+\tau})-g(\alpha^{t}))}.
\]
By the Weil–Carlitz–Uchiyama bound (see \cite{20,1}, among others),
\[
    |C_{f,g}|\le (d-1)\sqrt{q}+1
\]
whenever $f\neq g$ or $\tau\neq 0$. This method of construction can be generalized from finite fields to arbitrary function fields.

A family of binary sequences with length $2^{n}+1$ ($p^{n}+1$), cardinality $2^{n}-1$ ($p^{n}-1$), and correlation bounded by $2^{(n+2)/2}$ ($4+\lfloor2 p^{\frac{n}{2}}\rfloor$) was constructed via cyclotomic function fields (see \cite{2,9}). Using the framework of Artin–Schreier extensions of function fields, researchers subsequently derived sequences that combine low correlation with large linear complexity. In particular, Xing \textit{et al.} proposed a general construction over function fields of even characteristic; see \cite{5}. In \cite{10}, Hu \textit{et al.} gave another construction based on Kummer extensions of function fields. Moreover, Luo \textit{et al.} \cite{20} developed sequences with high nonlinear complexity obtained from rational function fields and cyclotomic function fields. Later, Hu \textit{et al.} \cite{22} introduced a further family of $p$-ary pseudorandom sequences arising from elliptic curves over finite fields and investigated their periods, linear complexities, aperiodic correlations, and other properties. Although many such constructions of binary sequences with good parameters are known, the achievable sequence periods over any fixed finite field are still quite limited. More recently, Jin \textit{et al.} \cite{3} and Liu \textit{et al.} \cite{21} independently constructed two families of binary sequences by exploiting the cyclic structure of rational points on elliptic curves over finite fields of even and odd characteristic, respectively.

In this paper, we use exponential sums over algebraic function fields, together with results of Deligne \cite{32}, Weil \cite{31}, and Bombieri \cite{33}, to derive properties of the periods, linear complexity, and periodic correlation of a family of $p$-ary sequences constructed from function fields.

The remainder of the paper is organized as follows. In Section \ref*{sec:2}, we recall basic notions on sequences, extension theory of function fields, exponential sums over function fields, and cyclic elliptic function fields, which will be needed throughout the paper. Section \ref*{sec:3} presents a theoretical framework for constructing $p$-ary sequences with low correlation and large linear complexity via general function fields over finite fields. In Sections \ref*{sec:4} and \ref*{sec:5}, we give two explicit constructions, one based on rational function fields and the other on cyclic elliptic function fields, respectively.
   \section{Preliminaries }
	\label{sec:2}
	In this section, we list some preliminaries on the results of exponential sums over function fields and cyclic elliptic function fields.

\subsection{The linear complexity and correlation}
Linear complexity and correlation are two fundamental metrics for assessing the pseudo-random properties of sequences.
\begin{definition}
    The linear complexity of a periodic sequence $S=\{s_i\}_{i=0}^{\infty}$ (or of a finite sequence $S=\{s_i\}_{i=0}^{N-1}$) over $\mathbb{F}_p$ is defined as the smallest nonnegative integer $\mathrm{L}$ such that there exist coefficients $\lambda_0,\lambda_1,\dots,\lambda_{\mathrm{L}}\in\mathbb{F}_p$ with $\lambda_0=1$ and $\lambda_{\mathrm{L}}\neq 0$ satisfying
\[
\lambda_0 s_{i+\mathrm{L}}+\lambda_1 s_{i+\mathrm{L}-1}+\cdots+\lambda_{\mathrm{L}} s_i = 0
\]
for every $i\ge 0$ (or for all $0\le i\le N-\mathrm{L}$ in the finite case). The linear complexity of $S$ is denoted by $\mathrm{LC}(S)$.
The researches of the deviation of $\mathrm{L}_{n}(S)$ from $N/2$ is also considered
\end{definition}
For $1\leq n\leq N-1$, let $\mathrm{LC}_{n}(S)$ represent the linear complexity of the truncated sequence $S_{n}=\{s_{i}\}_{i=0}^{n-1}$. The sequence $\{\mathrm{LC}_{n}(S)\}_{n=1}^{N-1}$ is referred to as the linear complexity profile of $S$.

Now we introduce the definition of the $d-$perfect for a sequence, see \cite{4}.
\begin{definition}
    If $d$ is a positive integer, then a sequence $S$ of elements of $\mathbb{F}_{q}$ is called $d-$perfect if
    \begin{displaymath}
        |\mathrm{LC}_{n}(S)-n|\leq d\ \ \text{for all}\ n\geq 1.
    \end{displaymath}
    A $1-$perfect sequence is also called perfect. A sequence is called almost perfect if it is $d-$perfect for some $d$.
\end{definition}

\begin{definition}
   Let $S_{1}=\{s_{i,1}\}_{i=0}^{\infty}$ and $S_{2}=\{s_{i,2}\}_{i=0}^{\infty}$ be two sequences over $\mathbb{F}_{p}$, each with period $T$. The periodic correlation $C_{S_{1},S_{2}}(\tau)$ at shift $\tau$ is given by
\begin{displaymath}
    C_{S_{1},S_{2}}(\tau)=\sum_{i=0}^{T-1}\omega_{p}^{\,s_{1,i+\tau}-s_{2,i}}.
\end{displaymath}
where $\omega_{p}$ is denoted by the $p$-th root of the unity $e^{2\pi i/p}$.
    \end{definition}

\subsection{Nonlinear Complexity}
 We denote the polynomial ring formed by $L$ variables over finite field by $\mathbb{F}_{q}$ $\mathbb{F}_{q}[x_{1},\ldots,x_{\mathrm{L}}]$. For a monomial $x_{1}^{e_{1}}\cdots x_{\mathrm{L}}^{e_{\mathrm{L}}}$, we denote by $I =(e_{1},\cdots,e_{\mathrm{L}}) \in \mathbb{Z}^L_{\geq 0}$, with its total degree calculated as $\mathrm{wt}(I) = \sum_{i=1}^{L} e_{i}$. Any element in $\mathbb{F}_{q}[x_{1},...,x_{\mathrm{L}}] $ with a total degree not exceeding $m$ can be expressed as
$\sum_{\mathrm{wt}(I) \leq m}a_Ix_{1}^{e_{1}}\cdots x_{\mathrm{L}}^{e_{\mathrm{L}}},$
where all coefficients $a_{I}$ are elements of $\mathbb{F}_{q}$.

For $\mathbf{s}=(s_{0},s_{1},\cdots,s_{N-1})\neq \mathbf{0}$ defined over $\mathbb{F}_{q}$, the multivariate polynomial $\Phi(x_{1},\cdots,x_{\mathrm{L}})$ generates $\mathbf{s}$ if $s_{i+\mathrm{L}}=\Phi(s_{i},s_{i+1},\cdots,s_{i+\mathrm{L-1}})$ for $0\leq i\leq N-1-\mathrm{L}$.
Fix an integer $m$, we denote the $m$th-order nonlinearity by $\mathrm{NL}_{m}(\mathbf{s})$ which associates a minimal positive integer $\mathrm{L}$ where a multivariate polynomial $\Phi(x_{1},\cdots,x_{\mathrm{L}})$ with total degree $\leq m $ can be found to generate the sequence $\mathbf{s}$. Particularly, if each component of $\mathbf{s}$ is zero , $\mathrm{NL}_{m}(\mathbf{s})$ is defined as zero. For $1\leq n\leq N$, $\pi_{n}(\mathbf{s})$ is used to denote the truncated sequence $(s_{0}, s_{1}, \cdots, s_{n-1})$ which is derived from $\mathbf{s}$ and the $m$th-order nonlinear complexity of this sequence is denoted by $\mathrm{NL}_{m}(\mathbf{s},n)$.

\subsection{Function fields and exponential sums}
Let $F/\mathbb{F}_{q}$ be an algebraic function field of genus $g(F)$ over the constant field $\mathbb{F}_{q}$. A place is defined as a maximal ideal in a valuation ring, and a divisor is a finite formal linear combination of such places. We write $\mathbb{P}_{F}$ for the set of all places of $F$ and $\mathbb{D}_{F}$ for the set of all divisors on $F$.

For a place $P$ of $F$, let $O_{P}$ be the associated valuation ring, and define the residue class field $F_{P}=O_{P}/P$. The degree of $P$ is given by $[F_{P}:\mathbb{F}_{q}]$. Let $P^{1}(F)$ denote the set of $\mathbb{F}_{q}$-rational places of $F$ and put $N(F)=|P^{1}(F)|$. The classical Hasse–Weil bound, sharpened by Serre, states that
\[
   |N(F)-q-1|\leq g(F)\cdot\lfloor2\sqrt{q}\rfloor,
\]
Here, $\lfloor x\rfloor$ denotes the integer part of $x\in\mathbb{R}$. This estimate is commonly used, for example, to obtain bounds on correlation measures and linear complexity in \cite{3,21,5}, among other works.

Let $v_{P}$ be a normalized discrete valuation of $F$. For a divisor $G=\sum_{P\in\mathbb{P}_{F}}m_{P}P\in\mathbb{D}_{F}$, we write $v_{P}(G)=m_{P}$, and its support is
\[
   \mathrm{Supp}(G)=\{P\in\mathbb{P}_{F}\mid v_{P}(G)\neq 0\}.
\]
The degree of $G$ is defined by
\[
   \deg(G)=\sum_{P\in\mathbb{P}_{F}}v_{P}(G)\deg(P).
\]
A divisor $G=\sum_{P\in\mathbb{P}_{F}}v_{P}(G)P$ is called effective if $v_{P}(G)\geq 0$ for every $P\in\mathrm{Supp}(G)$.

For any rational function $f\in F$, write $(f)$, $(f)_{0}$, and $(f)_{\infty}$ for the principal, zero, and pole divisors of $f$, respectively. Given a nonnegative divisor $G$, the associated Riemann–Roch space is
\[
   \mathcal{L}(G):=\{f\in F\setminus\{0\}\mid(f)+G\geq 0\}\cup\{0\},
\]
which is a finite-dimensional $\mathbb{F}_{q}$-vector space. We denote its dimension by $\ell(G)$. If $\deg(G)\geq 2g(F)-1$, then the Riemann–Roch theorem yields
\[
   \ell(G)=\deg(G)+1-g(F).
\]

An element $z \in F$ is called degenerate if there exist some $\alpha \in \mathbb{F}_q$ and $h \in F$ such that $z = \alpha + h^{p} - h$; if no such representation exists, then $z$ is called non-degenerate. The next lemma gives a characterization of when an element is non-degenerate.

\begin{lemma}[{\cite[Lemma 2.2]{26}}]
   If there is a place $Q$ of $F$ for which $v_{Q}(z)<0$ and this integer is relatively prime to $p$, then $z$ is non-degenerate.
    \label{88}
\end{lemma}

Let $\mathrm{Aut}(F/\mathbb{F}_{q})$ denote the automorphism group of $F$ over $\mathbb{F}_{q}$, that is,
\[
    \mathrm{Aut}(F/\mathbb{F}_{q})=\{\sigma : F\to F \mid \sigma\ \text{is an $\mathbb{F}_{q}$-automorphism of}\ F\}.
\]
In what follows, we investigate the natural action of $\mathrm{Aut}(F/\mathbb{F}_{q})$ on the collection of places of $F$. Given a place $P$ of $F$ and an automorphism $\sigma \in \mathrm{Aut}(F/\mathbb{F}_{q})$, the image $\sigma(P)$ is itself a place of $F$. This observation yields the following results.
\begin{lemma}(see \cite{7})
    Let $\sigma\in \mathrm{Aut}(F/\mathbb{F}_{q})$ be an automorphism, $P\in\mathbb{P}_{F}$ a place, and $f\in F$ a function. Then:
    \begin{enumerate}
        \item $\deg(\sigma(P))=\deg(P)$;
        \item $v_{\sigma(P)}(\sigma(f))=v_{P}(f)$;
        \item if $v_{P}(f)\geq 0$, then $\sigma(f)(\sigma(P))=f(P)$.
    \end{enumerate}
    \label{1}
    \end{lemma}

  For $f \in F$, define
\begin{displaymath}
    \begin{cases}
        m_{P}(f)=0 & \text{if } f\in\mathcal{O}_{P},\\[4pt]
        m_{P}(f)=\text{the order of the pole of } f \text{ at } P & \text{if } f\notin\mathcal{O}_{P},\\[4pt]
        m_{P}^{*}(f)=\min\{\,m_{P}(f-g)\mid g\in\mathscr{P}(F)\,\},
    \end{cases}
\end{displaymath}
where the map $\mathscr{P}$ is defined by $\mathscr{P}(x)=x - x^{p}$ for all $x\in F$, and $m_{P}^{*}$ is considered as a function on the quotient space $F/\mathscr{P}(F)$.

Evidently, $m_{P}^{*}(f)=0$ holds precisely when $f$ is degenerate. In accordance with Artin’s terminology, we refer to $m_{P}^{*}(f)$ as the reduced order of the pole of $f$ at the place $P$.

The following theorem on exponential sums over function fields, first proved by Weil, Deligne, and Bombieri, is highly valuable for deriving bounds on linear complexity and correlation.
\begin{theorem}[see \cite{31,32,33}]
    For any non-degenerate $f\in F$, let $A=P^1(F)\setminus\mathrm{Supp}((f)_{\infty})$, then
    \begin{displaymath}
        \left|\sum_{P\in A}\omega_{p}^{\mathrm{Tr}\left(f(P)\right)}\right|\leq\\
        \left(2g(F)-2+\sum_{u\in\mathrm{Supp}((f)_{\infty})}(m^{*}_{u}(f)+1)\deg(u)\right)\sqrt{q}.
         \label{2}    \end{displaymath}
    \end{theorem}
We now present several results on exponential sums that are used in \cite{22} \textit{etc.}, and we observe that these are in fact special cases of Theorem \ref{2}.
\begin{remark}
\begin{enumerate}
    \item In the case where all poles of $f$ have degree $1$, we obtain the special situation considered by Bombieri in \cite{33}, namely
    \begin{displaymath}
        \left|\sum_{P\in A}\omega_{p}^{\mathrm{Tr}\left(f(P)\right)}\right|\leq
        \left(2g(F)-2+\#\mathrm{Supp}((f)_{\infty})+\deg((f)_{\infty})\right)q^{1/2}.
    \end{displaymath}
    \item If $F$ is an elliptic function field and each pole of $f$ has degree $1$, then we obtain the estimate
    \begin{displaymath}
        \left|\sum_{P\in A}\omega_{p}^{\mathrm{Tr}\left(f(P)\right)}\right|\leq\left(\#\mathrm{Supp}((f)_{\infty})+\deg((f)_{\infty})\right)q^{1/2}.
    \end{displaymath}
    Furthermore, if $f\in F$ has a unique pole $Q$ with $v_{Q}(f)=1$, then
     \begin{displaymath}
        \left|\sum_{P\in A}\omega_{p}^{\mathrm{Tr}\left(f(P)\right)}\right|\leq2\deg((f)_{\infty})q^{1/2}.
    \end{displaymath}
    These bounds were applied by Hu in \cite{22}.
    \end{enumerate}
\end{remark}

   \subsection{Elliptic Function Fields}
An elliptic function field $E/\mathbb{F}_{q}$ is an algebraic function field of genus one defined over $\mathbb{F}_{q}$.

We use the symbol $\oplus$ for the group operation on the set of $\mathbb{F}_{q}$-rational points of the elliptic function field $E/\mathbb{F}_{q}$. We call an elliptic function field cyclic if its rational points form a cyclic group under $\oplus$. The automorphism group of an elliptic curve $E$ over $\mathbb{F}_{q}$ can be written as a semidirect product $T_{E}\rtimes\mathrm{Aut}(E,O)$, where $\mathrm{Aut}(E,O)$ is the group of $\mathbb{F}_{q}$-automorphisms of $E$ that fix the point at infinity $O$, and $T_{E}$ is the subgroup of translations, which can be canonically identified with $E$ itself (see \cite{23}).

Let $T_{E}=\{\sigma_{P} : P\in E(\mathbb{F}_{q})\}$ denote the translation group of $E$, where for each $P$ the automorphism $\sigma_{P}$ is defined by $\sigma_{P}(Q)=P\oplus Q$ for every rational place $Q$. Let $F$ be the subfield of $E$ fixed by the action of $T_{E}$, that is,
\begin{displaymath}
    F=E^{T_{E}}=\{z\in E : \tau(z)=z\ \text{for all}\ \tau\in T_{E}\},
\end{displaymath}
where $\tau(z)(P)=z(\tau^{-1}(P))$ for every $P\in E(\mathbb{F}_{q})$ and $\tau\in T_{E}$.

By the Galois theory, the extension $E/F$ is Galois with Galois group $\mathrm{Gal}(E/F)\simeq T_{E}$. Consequently, there exists a cyclic elliptic function field $E/\mathbb{F}_{q}$ with exactly $q+1+t$ rational places. In this situation, the translation group $T_{E}$ is a cyclic group of order $q+1+t$ generated by $\sigma_{P}$. For each $0\leq j\leq q+t$, define $P_{j}=\sigma_{P}^{j}(O)=[j]P$, which is a rational place of $E$; then the set $\{P,P_{1},\ldots,P_{q+t}\}$ is precisely the set of all rational places of $E$.

\begin{lemma}
     Let $q=p^{n}$. Let $t$ be an integer satisfying one of the following conditions:
     \begin{enumerate}
         \item $|t|\leq 2\sqrt{q}$ and $\gcd(t,p)=1$;
         \item $t=0$, $n$ is odd or $n$ is even with $q\not\equiv -1 \mod 4$;
          \item $t=\pm\sqrt{q}$ if $n$ is even and $p\not\equiv 1 \mod 3$; $t=\pm p^{(n+1)/2}$ if $n$ is odd and $p=2,3$.
          \end{enumerate}
Then there exists a cyclic elliptic function field $E/\mathbb{F}_{q}$ with $q+1+t$ rational places.
          \label{222}
 \end{lemma}

\begin{lemma}
Let $E$ be a cyclic elliptic function field and let $P$ be a generator of $\mathbb{P}^{1}_{E}$. Then, for any place $Q$ in $\mathbb{P}_{E}$ of degree $d$ satisfying $\gcd(d,q+1+t)=1$, and for any fixed positive integer $i$, the places $\sigma_{P}^{i}(Q),\sigma_{P}^{i+1}(Q),\ldots,\sigma_{P}^{i+q+t}(Q)$ are all distinct from one another.
\end{lemma}

Let $d$ be a positive divisor that is coprime to $q+1+t$, and let $B_d$ denote the number of places of degree $d$. Then the preceding lemma shows that there are precisely $r = B_d/(q+1+t)$ places of degree $d$ of $F = E^{T_E}$ that split completely in the extension $E/F$.

	\section{ The general construction of the pseudorandom sequences }
	\label{sec:3}
   In this section, we generalize the constructions from \cite{5,22} to sequences defined over finite fields of odd characteristic. We also introduce some notation that will be used throughout the section.
\begin{enumerate}
    \item $P$ is a rational place of $F$;
    \item $\sigma$ is an automorphism in $\mathrm{Aut}(F/\mathbb{F}_{q})$;
    \item $n$ is the smallest positive integer such that $\sigma^{n}(P)=P$ that is, $\sigma^{n}(P)=P$ and $\sigma^{i}(P)\neq P$ for every $1\leq i\leq n-1$.
\end{enumerate}
Set $P_i = \sigma^{i}(P)$ for every integer $i\in\mathbb{Z}$, with $P_0 = P$. Then we have $P_{j+n} = P_j$ for all $j\in\mathbb{Z}$, and for any fixed $\ell\in\mathbb{Z}$, the places
\[
P_{\ell}, P_{\ell+1}, \ldots, P_{\ell+n-1}
\]
 are $n$ pairwise distinct rational places. For any element $z \in F$ satisfying $v_{P_i}(z)\ge 0$, we associate the $p$-ary sequence
\[
\mathbf{s}_z = (\mathrm{Tr}(z(P_{0})), \mathrm{Tr}(z(P_{1})),\cdots,\mathrm{Tr}(z(P_{n-1}))).
\]
It follows that $n$ is a multiple of the period of $\mathbf{s}_z$.

We begin by studying the least period of the sequence.

\begin{proposition}
   Let $z\in F$ be an element with $v_{P_i}(z)\ge 0$. Suppose that $Q$ is the only pole of $z$, that $\gcd(v_Q(z),p)=1$, and that the places $Q,\sigma(Q),\ldots,\sigma^{n-1}(Q)$ are all distinct. If $d=\deg(z)_\infty$ and
\[
n> \left(2g(F)-2+2\bigl(m^{*}_{Q}(z)+1\bigr)d\right)\sqrt{q},
\]
then the sequence $\mathbf{s}_z$ has period exactly $n$.
    \begin{proof}
        Assume, for contradiction, that the least period of $\mathbf{s}_{z}$ is $k$ with $0<k<n$. Consider the rational function $f=z-\sigma^{k}(z)$. Since $Q$ is the only pole of $z$, Lemma \ref{1} implies that $\sigma^{k}(Q)$ is the only pole of $\sigma^{k}(z)$. Hence, the poles of $f$ are precisely $Q$ and $\sigma^{k}(Q)$, and
        \[
            \deg(f)_{\infty}=2v_{Q}(z)d.
        \]
        For any $i\ge 0$, we compute
        \begin{displaymath}
        \begin{split}
            f(P_{i+k})&=z(P_{i+k})-\sigma^{k}(z)(P_{i+k})\\
            &=z(P_{i+k})-z(\sigma^{-k}(P_{i+k}))\\
            &=z(P_{i+k})-z(P_{i}),
        \end{split}
        \end{displaymath}
        valid for all $i\ge 0$. Consequently,
        \begin{displaymath}
            \mathrm{Tr}(f(P_{i+k}))=\mathrm{Tr}(z(P_{i+k})-z(P_{i}))=0.
        \end{displaymath}
        Therefore,
        \begin{displaymath}
            \left|\sum_{i=0}^{n-1}\omega_{p}^{\mathrm{Tr}(f(P_{i+k}))}\right|=n.
        \end{displaymath}
        On the other hand, by Lemma \ref{2} we have
        \begin{displaymath}
        \begin{split}
            \left|\sum_{i=0}^{n-1}\omega_{p}^{\mathrm{Tr}(f(P_{i+k}))}\right|
            &\le \left(2g(F)-2+\bigl(m^{*}_{Q}(f)+m^{*}_{\sigma^{k}(Q)}(f)+2\bigr)d\right)\sqrt{q}\\
            &=\left(2g(F)-2+2(m^{*}_{Q}(z)+1)d\right)\sqrt{q}.
        \end{split}
        \end{displaymath}
        Applying Lemma \ref{2} thus yields
        \[
            n\le \left(2g(F)-2+2(m^{*}_{Q}(z)+1)d\right)\sqrt{q},
        \]
        contradicting our assumption on $n$. This proves the assertion.
    \end{proof}
    \label{3}
\end{proposition}

The linear complexity has likewise been examined in \cite{5,22,10,9,3} \textit{et al}. In this work, we analyze the linear complexity of a $p$-ary sequence generated from exponential sums defined over general algebraic function fields.

We first treat the case $p=q$. The basic idea follows essentially the same argument as in \cite{22}.
\begin{theorem}
  Let $z\in F$ be such that $v_{P_{i}}(z)\geq 0$. Suppose that $Q$ is the unique pole of $z$ and that $\gcd(v_{Q}(z),p)=1$; furthermore, assume that the places $Q,\sigma(Q),\ldots,\sigma^{n-1}(Q)$ are all distinct. Then the linear complexity of the sequence $\mathbf{s}_{z}$ satisfies
\[
\mathrm{L}(\mathbf{s}_{z})\geq\frac{n-m^{*}_{Q}(z)d}{m^{*}_{Q}(z)d+1}.
\]
\begin{proof}
   Assume that the linear complexity of $\mathbf{s}_{z}$ is $\mathrm{L}=\mathrm{L}(\mathbf{s}_{z})$ with $0\leq \mathrm{L}<n$. Then there exist elements $\lambda_{0},\lambda_{1},\ldots,\lambda_{\mathrm{L}}\in\mathbb{F}_{p}$ such that $\lambda_{0}=1$, $\lambda_{\mathrm{L}}\neq 0$, and
\[
\lambda_{0}s_{i+\mathrm{L}}+\lambda_{1}s_{i+\mathrm{L}-1}+\cdots+\lambda_{\mathrm{L}}s_{i}=0
\]
for all $i\geq 0$.

Now consider the rational function
\[
f=\lambda_{0}\sigma^{-\mathrm{L}}(z)+\lambda_{1}\sigma^{-\mathrm{L}+1}(z)+\cdots+\lambda_{\mathrm{L}}z.
\]
By Lemma \ref{88}, the function $f$ is non-degenerate. Furthermore, for each $0\leq i\leq \mathrm{L}+1$, the unique pole of $\sigma^{-i}_{P}(z)$ is $\sigma^{-i}(Q)$. Hence the poles of $f$ are precisely the points $Q,\sigma^{-1}(Q),\ldots,\sigma^{-\mathrm{L}}(Q)$. It follows that $f\neq 0$, and therefore $\deg(f)=(\mathrm{L}+1)d$. For every $i\geq 0$, we have
        \begin{displaymath}
        \begin{split}
            f(P_{i})&=\lambda_{0}\sigma^{-\mathrm{L}}(z)(P_{i})+\lambda_{1}\sigma^{-\mathrm{L}+1}(z)(P_{i})+\cdots+\lambda_{\mathrm{L}}z(P_{i})\\
            &=\lambda_{0}z(\sigma^{\mathrm{L}}(P_{i}))+\lambda_{1}z(\sigma^{\mathrm{L}-1}(P_{i}))+\cdots+\lambda_{\mathrm{L}}z(P_{i})\\
            &=\lambda_{0}z(P_{i+\mathrm{L}})+\lambda_{1}z(P_{i+\mathrm{L}-1})+\cdots+\lambda_{\mathrm{L}}z(P_{i}).
        \end{split}
    \end{displaymath}
    Consequently, we obtain
    \begin{displaymath}
        \begin{split}
            \mathrm{Tr}(f(P_{i}))&=\lambda_{0}\mathrm{Tr}(z(P_{i+\mathrm{L}}))+\lambda_{1}\mathrm{Tr}(z(P_{i+\mathrm{L}-1}))+\cdots+\lambda_{\mathrm{L}}\mathrm{Tr}(z(P_{i}))\\
            &=s_{i+\mathrm{L}}+d_{1}s_{i+\mathrm{L}-1}+\cdots+\lambda_{\mathrm{L}}s_{i}=0.\\
        \end{split}
    \end{displaymath}
Thus, the places $[i]P$ are zeros of $f$ for $0 \leq i \leq n-1-\mathrm{L}$. Denote by $\mathcal{Z}(f)$ and $\mathcal{N}(f)$ the sets of zeros and poles of $f$, respectively. Then
\begin{displaymath}
\begin{split}
    \sum_{P\in\mathcal{Z}(f)} v_{P}(f)\deg(P)
    &= -\sum_{P\in\mathcal{N}(f)} v_{P}(f)\deg(P)\\
    &= (\mathrm{L}+1)m^{*}_{Q}(z)d.
\end{split}
\end{displaymath}
Consequently,
\begin{displaymath}
    n-\mathrm{L}\leq (\mathrm{L}+1)m^{*}_{Q}(z)d,
\end{displaymath}
which can be rewritten as
\begin{displaymath}
 \mathrm{L}(\mathbf{s}_{z})\geq\frac{n-m^{*}_{Q}(z)d}{m^{*}_{Q}(z)d+1}.
\end{displaymath}
This yields the desired assertion.    
        \end{proof}
\end{theorem}

In what follows, we restrict our attention to the case $q>p$, where it is necessary to utilize the estimation for exponential sums over function fields provided in Theorem~\ref{2}.
\begin{theorem}
   Let $z\in F$ be such that $v_{P_{i}}(z)\geq 0$. Assume that $Q$ is the unique pole of $z$ with $\gcd(v_{Q}(z),p)=1$, and that the places $Q,\sigma(Q),\dots,\sigma^{n-1}(Q)$ are mutually distinct. Then the linear complexity of $\mathbf{s}_{z}$ satisfies
\[
\mathrm{L}(\mathbf{s}_{z})\geq\frac{n-(2g(F)-2)\sqrt{q}}{(m^{*}_{Q}(z)+1)d\sqrt{q}}-1.
\]
    \begin{proof}
      Suppose that the linear complexity of $\mathbf{s}_{z}$ is $\mathrm{L}=\mathrm{L}(\mathbf{s}_{z})$, with $0\leq \mathrm{L}<n$. Then there exist elements $\lambda_{0},\lambda_{1},\ldots,\lambda_{\mathrm{L}}\in\mathbb{F}_{p}$ such that $\lambda_{0}=1$ and $\lambda_{\mathrm{L}}\neq 0$, and they satisfy
\[
\lambda_{0}s_{i+\mathrm{L}}+\lambda_{1}s_{i+\mathrm{L}-1}+\cdots+\lambda_{\mathrm{L}}s_{i}=0
\]
for all integers $i\geq 0$.

Now consider the rational function
\[
f = \lambda_{0}\sigma^{-\mathrm{L}}(z) + \lambda_{1}\sigma^{-\mathrm{L}+1}(z) + \cdots + \lambda_{\mathrm{L}}z,
\]
which is non-degenerate by Lemma \ref{88}. Furthermore, for each integer $0 \leq i \leq \mathrm{L}+1$, the pole of $\sigma^{-i}_{P}(z)$ lies at the point $\sigma^{-i}(Q)$. Hence, the set of poles of $f$ is exactly
\[
Q, \sigma^{-1}(Q), \ldots, \sigma^{-\mathrm{L}}(Q).
\]
It follows that $f$ is nonzero, and its degree is given by $\deg(f) = (\mathrm{L}+1)d$. For every $i \geq 0$, we obtain
        \begin{displaymath}
        \begin{split}
            f(P_{i})&=\lambda_{0}\sigma^{-\mathrm{L}}(z)(P_{i})+\lambda_{1}\sigma^{-\mathrm{L}+1}(z)(P_{i})+\cdots+\lambda_{\mathrm{L}}z(P_{i})\\
            &=\lambda_{0}z(\sigma^{\mathrm{L}}(P_{i}))+\lambda_{1}z(\sigma^{\mathrm{L}-1}(P_{i}))+\cdots+\lambda_{\mathrm{L}}z(P_{i})\\
&=\lambda_{0}z(P_{i+\mathrm{L}})+\lambda_{1}z(P_{i+\mathrm{L}-1})+\cdots+\lambda_{\mathrm{L}}z(P_{i}).
            \end{split}
            \end{displaymath}
            It follows that
            \begin{displaymath}
            \begin{split}
                \mathrm{Tr}(f(P_{i}))&=\lambda_{0}\mathrm{Tr}(z(P_{i+\mathrm{L}}))+\lambda_{1}\mathrm{Tr}(z(P_{i+\mathrm{L}-1}))+\cdots+\lambda_{\mathrm{L}}\mathrm{Tr}(z(P_{i}))\\
&=s_{i+\mathrm{L}}+d_{1}s_{i+\mathrm{L}-1}+\cdots+\lambda_{\mathrm{L}}s_{i}=0.\\
                \end{split}
            \end{displaymath}
            Hence, we have 
            \begin{displaymath}
                \left|\sum^{n-1}_{i=0}\omega_{p}^{\mathrm{Tr}(f(P_{i}))}\right|=n.
                \end{displaymath}
            On the other hand, by Lemma \ref{2}, we have
            \begin{displaymath}
                \left|\sum^{n-1}_{i=0}\omega_{p}^{\mathrm{Tr}(f(P_{i}))}\right|\leq \left(2g(F)-2+\sum^{\mathrm{L}}_{i=0}(v_{\sigma^{i}(Q)}^{*}(f)+1)d\right)\sqrt{q}.
                \end{displaymath}
            Then we have
            \begin{displaymath}
\mathrm{L}(\mathbf{s}_{z})\geq\frac{n-(2g(F)-2)\sqrt{q}}{(m^{*}_{Q}(z)+1)d\sqrt{q}}-1.
            \end{displaymath}
\end{proof}
                \label{10}
\end{theorem}

\begin{corollary}
   Let $z\in F$ be such that $v_{P_{i}}(z)\geq 0$. Assume that $Q$ is the unique pole of $z$ with $\gcd(v_{Q}(z),p)=1$, and that the places $Q,\sigma(Q),\dots,\sigma^{n-1}(Q)$ are mutually distinct. Then we have the following estimation
   \begin{displaymath}
        |2\mathrm{L}(\mathbf{s}_{z})-n|\leq \frac{2n-\left(  2(2g(F)-2)+(n+2)(m^{*}_{Q}(z)+1)d\right)\sqrt{q}}{(m^{*}_{Q}(z)+1)d\sqrt{q}}.
    \end{displaymath}
\end{corollary}
From the preceding results, we immediately obtain a corollary concerning perfect sequences, showing that the construction can actually be perfect when working over rational function fields or cyclic elliptic function fields and we shall list them in the last section.

We now examine the auto- and cross-correlation properties of the sequences defined above.
\begin{theorem}
     Let $z_i \in F$ be two elements of $F$ such that $v_{P_j}(z_i)\ge 0$ for all $0 \le j \le n$. Assume that for each $i$, the element $z_i$ has a unique pole $Q_i$ with $\gcd\big(v_{Q_i}(z_i),p\big)=1$, and let $d_i = \deg\big((z_i)_\infty\big)$ for $1 \le i \le 2$ (we allow the case $z_1 = z_2$). Suppose further that $z_1 + \sigma^{-\tau}(z_2)$ is non-degenerate for some $\tau$. Then the correlation at shift $\tau$ satisfies
    \begin{displaymath}
    \left|C_{\mathbf{s}_{z_{1}},\mathbf{s}_{z_{2}}}(\tau)\right|\leq    \begin{cases}
 \left(2g(F)-2+\sum^{2}_{i=1}(m^{*}_{Q_{i}}(z_{i})+1)d_{i}\right)\sqrt{q},&\text{if}\ 0<\tau< n;\\
 \left(2g(F)-2+\sum^{2}_{i=1}(m^{*}_{Q_{i}}(z_{i})+1)d_{i}\right)\sqrt{q},&\text{if}\ \tau=0,Q_{1}\neq Q_{2};\\       
 \left(2g(F)-2+(m^{*}_{Q_{1}}(z_{1})+1)d_{1}\right)\sqrt{q},&\text{if}\ \tau=0, Q_{1}=Q_{2},z_{1}\neq z_{2}.\\         
        \end{cases}
        \end{displaymath}
         \label{1989}
         \end{theorem}
    \begin{proof}
       The argument proceeds in the same way as in Proposition \ref{3}. Take the rational function $z = z_{1} + \sigma^{-\tau}(z_{2})$. It is straightforward to verify that $f$ has exactly two poles, namely $Q$ and $\sigma^{-\tau}(Q)$. By definition, we obtain
        \begin{displaymath}
        \begin{split}
             \left|C_{\mathbf{s}_{z_{1}},\mathbf{s}_{z_{2}}}(\tau)\right|&=\left|\sum^{n}_{i=1}\omega_{p}^{\mathrm{Tr}(z_{1}(P_{i}))+\mathrm{Tr}(z_{2}(P_{i+\tau}))}\right|=\left|\sum^{n}_{i=1}\omega_{p}^{\mathrm{Tr}(z_{1}(P_{i}))+\mathrm{Tr}(\sigma^{-\tau}(z_{2})(P_{i}))}\right|\\
             &=\left|\sum^{n}_{i=1}\omega_{p}^{\mathrm{Tr}(z_{1}(P_{i})+\sigma^{-\tau}(z_{2})(P_{i}))}\right|=\left|\sum^{n}_{i=1}\omega_{p}^{\mathrm{Tr}(z(P_{i}))}\right|\\
              &\leq \left(2g(F)-2+(m^{*}_{Q_{1}}(f)+1)d_{1}+(m^{*}_{\sigma^{-\tau}(Q_{2})}(f)+1)d_{2}\right)\sqrt{q}\ \text{(By Lemma \ref{2})}\\
              &= \left(2g(F)-2+\sum_{i=1}^{2}(m^{*}_{Q_{i}}(z_{i})+1)d_{i}\right)\sqrt{q}\\
              \end{split}
             \end{displaymath}
    for $0<\tau<n$, and if $\tau=0$ and $Q_{1}\neq Q_{2}$, then we have
    \begin{displaymath}
         \left|C_{\mathbf{s}_{z_{1}},\mathbf{s}_{z_{2}}}(\tau)\right|\leq 
          \left(2g(F)-2+\sum_{i=1}^{2}(m^{*}_{Q_{i}}(z_{i})+1)d_{i}\right)\sqrt{q}.
          \end{displaymath}
        Applying Lemma \ref{2} once more yields the required conclusions. In the case where $\tau=0, Q_{1}=Q_{2}$, and $z_{1}\neq z_{2}$, the function $f$ has a single pole at $Q_{1}$. Combining Lemma \ref{2} with the condition $\gcd(v_{Q_{i}}(z_{i}),p)=1$ for $i=1,2$ again leads to the desired result.
         \end{proof}

We now study the distribution of $r$-patterns in the sequences constructed above. To begin, we recall the orthogonality of exponential sums.
\begin{lemma}
    \begin{displaymath}
        \sum^{p-1}_{d=0}\omega_{p}^{ad}=\begin{cases}
            p,&\text{if}\ p\mid a;\\
            0,&\text{otherwise}.
        \end{cases}
    \end{displaymath}
    \label{666}
\end{lemma}
The following theorem presents the estimation of $r$-patterns of the sequences we construct.
\begin{theorem}
    Let $z\in F$ satisfy $v_{P_{i}}(z)\geq 0$. Suppose that $Q$ is the unique pole of $z$, that the places $Q,\sigma(Q),\ldots,\sigma^{n-1}(Q)$ are all mutually distinct, and that $\gcd(p,v_{Q}(z))=1$. Let $(\overline{z},\overline{t})\in\mathbb{F}_{p}^{r}\times[0,n)^{r}$ be an $r$-pattern of the sequence $\mathbf{s}_{z}$, where $\overline{z}=(z_{1},z_{2},\ldots,z_{r})$ and $\overline{t}=(t_{1},t_{2},\ldots,t_{r})$ with $0\leq t_{1}<t_{2}<\cdots<t_{r}<T$. We define $N_{(\overline{z},\overline{t})}$ as the number of indices $i_{0}$, $0\leq i_{0}<T$, such that $s_{i_{0}+t_{j}}=z_{j}$ holds for every $j=1,2,\ldots,r$. Then we have
    \begin{displaymath}
        \left|N_{(\overline{z},\overline{t})}-\frac{n}{p^{r}}\right|\leq (2g(F)-2)\sqrt{q}+dr\sqrt{q}(m^{*}_{Q}(z)+1)(1-\frac{1}{p}).   
           \end{displaymath}
    \begin{proof}
       Using the orthogonality of the exponential sums established in Lemma \ref{666}, we obtain
        \begin{displaymath}
        \begin{split}
            N_{(\overline{z},\overline{t})}&=\sum^{n-1}_{i=0}\frac{1}{p}\sum^{p-1}_{d_{1}=0}\omega_{p}^{d_{1}(s_{i+t_{1}}-z_{1})}\cdots\frac{1}{p}\sum^{p-1}_{d_{r}=0}\omega_{p}^{d_{r}(s_{i+t_{r}}-z_{r})}\\
            &=\frac{1}{p^{r}}\sum^{p-1}_{d_{1},d_{2},\cdots,d_{r}=0}\omega_{p}^{-d_{1}z_{1}-\cdots-d_{r}z_{r}}\sum^{n-1}_{i=0}\omega_{p}^{d_{1}s_{i+t_{1}}+\cdots+d_{r}s_{i+t_{r}}}.\\
            \end{split}
    \end{displaymath}
Thus
\begin{displaymath}
\begin{split}
 N_{(\overline{z},\overline{t})}-\frac{n}{p^{r}}=\frac{1}{p^{r}}\sum^{p-1}_{(d_{1},d_{2},\cdots,d_{r})\neq (0,0,\cdots,0)}\omega_{p}^{-d_{1}z_{1}-\cdots-d_{r}z_{r}}\sum^{n-1}_{i=0}\omega_{p}^{d_{1}s_{i+t_{1}}+\cdots+d_{r}s_{i+t_{r}}}.
 \end{split}
\end{displaymath}
Thus, we have the upper bound
\begin{displaymath}
\begin{split}
 \left|N_{(\overline{z},\overline{t})}-\frac{n}{p^{r}}\right|\leq\frac{1}{p^{r}}\sum^{p-1}_{(d_{1},d_{2},\cdots,d_{r})\neq (0,0,\cdots,0)}\left|\sum^{n-1}_{i=0}\omega_{p}^{d_{1}s_{i+t_{1}}+\cdots+d_{r}s_{i+t_{r}}}\right|.
 \end{split}
\end{displaymath}
Now we examine the bound for $\left|\sum^{n-1}_{i=0}\omega_{p}^{d_{1}s_{i+t_{1}}+\cdots+d_{r}s_{i+t_{r}}}\right|$.

Without loss of generality, let $k$ denote the Hamming weight of the vector $(d_{1},\ldots,d_{r})$, where $1\leq k\leq r$ and suppose that $d_{1},\ldots,d_{k}\in\mathbb{F}_{p}^{*}$ while $d_{k+1}=\cdots=d_{r}=0$. Under this assumption, we have
\begin{displaymath}
\sum^{n-1}_{i=0}\omega_{p}^{d_{1}s_{i+t_{1}}+\cdots+d_{r}s_{i+t_{r}}}=\sum^{n-1}_{i=0}\omega_{p}^{d_{1}s_{i+t_{1}}+\cdots+d_{k}s_{i+t_{k}}}.
\end{displaymath}
Consider $f=d_{1}\sigma^{-t_{1}}(z)+\cdots+d_{k}\sigma^{-t_{k}}(z)$. By the condition, we have that for any $1\leq i\leq k$, the pole of $\sigma^{-t_{i}}(z)$ is $\sigma^{-t_{i}}(Q)$. Then $\sigma^{-t_{1}}(Q),\cdots,\sigma^{-t_{k}}(Q)$. Hence, $f$ is nonzero, and $\deg(f)=\sum^{k}_{i=1}m^{*}_{\sigma^{-t_{i}}(Q)}(\sigma^{-t_{i}}(z))d=km^{*}_{Q}(z)d$. Then we have
\begin{displaymath}
    \mathrm{Tr}(f(P_{i}))=d_{1}s_{i+t_{1}}+\cdots+d_{r}s_{i+t_{r}}.
    \end{displaymath}
By Lemma \ref{2}, the following inequality holds
\begin{displaymath}
    \left|\sum^{n-1}_{i=0}\omega_{p}^{d_{1}s_{i+t_{1}}+\cdots+d_{r}s_{i+t_{r}}}\right|\leq \left(2g(F)-2+\sum^{k}_{i=1}(m^{*}_{\sigma^{-t_{i}}(Q)}(\sigma^{-t_{i}}(z))+1)d\right)\sqrt{q}
    \end{displaymath}
Therefore, we have
\begin{displaymath}
\begin{split}
    &\sum_{(d_{1},\cdots,d_{r})\neq(0,\cdots,0)}\left|\sum^{n-1}_{i=0}\omega_{p}^{d_{1}s_{i+t_{1}}+\cdots+d_{r}s_{i+t_{r}}}\right|\\
    &\leq \sum^{r}_{k=1}\binom{r}{k}(p-1)^{k}\left(2g(F)-2+\sum^{k}_{i=1}(m^{*}_{\sigma^{-t_{i}}(Q)}(\sigma^{-t_{i}}(z))+1)d\right)\sqrt{q}\\
    &=(2g(F)-2)\sqrt{q}\sum_{k=1}^{r}\binom{r}{k}(p-1)^{k}+d\sqrt{q}(m^{*}_{Q}(z)+1)\sum_{k=1}^{r}\binom{r}{k}k(p-1)^{k}\\
    &=(2g(F)-2)\sqrt{q}p^{r}+dr\sqrt{q}(m^{*}_{Q}(z)+1)(p-1)p^{r-1}.
    \end{split}
    \end{displaymath}
Consider $f_{(\overline{z},\overline{t})}=N_{(\overline{z},\overline{t})}/n$; we have
\begin{displaymath}
    \left|f_{(\overline{z},\overline{t})}-\frac{1}{p^{r}}\right|\leq 
\frac{(2g(F)-2)\sqrt{q}}{n}+\frac{dr\sqrt{q}(m^{*}_{Q}(z)+1)(p-1)}{pn}.
    \end{displaymath}
For fixed $r,d,p,g(F)$, as $n\to\infty$ we have $f_{(\overline{z},\overline{t})}\to\frac{1}{p^{r}}$. In other words, the distributions of $r$-patterns in the sequences arising from any function fields become asymptotically uniform, thereby generalizing the result of Hu \textit{et al.} in \cite{22}.
    \end{proof}
    \end{theorem}

Nonlinear complexity is likewise a significant measure for sequences over function fields, and related results can be found in \cite{7,20,21}. 
\begin{theorem}
Let $\mathbb{F}_{q}$ be a finite field with $q = p^{h}$. Take an element $z \in F$ such that $v_{P_{i}}(z) \geq 0$. Assume that $z$ has $Q$ as its only pole, that $\gcd\bigl(v_{Q}(z),p\bigr) = 1$, and that the places $Q,\sigma(Q),\ldots,\sigma^{n-1}(Q)$ are mutually distinct. Then the $m$-th order nonlinear complexity of the sequence family $\mathbf{s}_{z}$ is given by
\begin{displaymath}
       \mathrm{NL}_{m}(\mathbf{s}_{z})\geq\frac{n-p^{h-1}m^{*}_{Q}(z)d}{1+mp^{h-1}m^{*}_{Q}(z)d}.
       \end{displaymath}
       \begin{proof}
   Assume that the $m$th-order nonlinear complexity of the family $\mathbf{s}_{z}$ is $\mathrm{L}$. Let $\Phi(x_{1},\ldots,x_{\mathrm{L}})=\sum_{\mathrm{wt}(I)\leq m}a_{I}x_{1}^{e_{1}}\cdots x_{\mathrm{L}}^{e_{\mathrm{L}}}$ be a nonzero multivariate polynomial in $\mathbb{F}_{q}[x_{1},\ldots,x_{\mathrm{L}}]$ with total degree at most $m$, satisfying $s_{j+\mathrm{L}}=\Phi(s_{j},\ldots,s_{j+\mathrm{L}-1})$ for every $0\leq j\leq n-1-\mathrm{L}$. Now consider the function
    \begin{displaymath}
        f:=\mathrm{Tr}(\tau^{-\mathrm{L}}(z))-\Phi(\mathrm{Tr}(z),\mathrm{Tr}(\tau^{-1}(z)),\ldots,\mathrm{Tr}(\tau^{-\mathrm{L}+1}(z))).
    \end{displaymath}
Observe that $\mathrm{Tr}(z)=z+z^{p}+\cdots+z^{p^{h-1}}$, and that $Q$ is the unique pole of $z$. By the non-Archimedean property of valuations, the pole divisor of $\mathrm{Tr}(z)$ is thus $(\mathrm{Tr}(z))_{\infty}=p^{h-1}v_{Q}(z)Q$.

  Moreover, since for each $0\leq i\leq \mathrm{L}$ the point $\sigma^{-i}(Q)$ is the only pole of $\mathrm{Tr}(\sigma^{-i}(z))$, we conclude that $\sigma^{-\mathrm{L}}(Q)$ is the sole pole of $f$. We now compare the degrees of the zero divisor and the pole divisor of $f$. Noting that $\sigma^{j}(P_{u})$ is not a pole of $\mathrm{Tr}(\sigma^{-i}(z))$ for any $0\leq j\leq n-1-\mathrm{L}$ and $0\leq i\leq\mathrm{L}$, we derive
   \begin{small}
    \begin{displaymath}
    \begin{split}
        f(\tau^{j}(P_{u}))&=\mathrm{Tr}(\sigma^{-\mathrm{L}}(z))(\tau^{j}(P_{u}))-\Phi(\mathrm{Tr}(z)(\sigma^{j}(P_{u})),\mathrm{Tr}(\tau^{-1}(z))(\sigma^{j}(P_{u})),\cdots,\mathrm{Tr}(\sigma^{-\mathrm{L}+1}(z))(\sigma^{j}(P_{u})))\\
        &=\mathrm{Tr}(z)(\sigma^{j+\mathrm{L}}(P_{u}))-\Phi(\mathrm{Tr}(z)(\sigma^{j}(P_{u})),\mathrm{Tr}(z)(\sigma^{j+1}(P_{u})),\cdots,\mathrm{Tr}(z)(\sigma^{j+\mathrm{L}-1}(P_{u})))\\
        &=s_{j+\mathrm{L}}-\Psi(s_{j},\cdots,s_{j+\mathrm{L}-1})=0\\
        \end{split}
    \end{displaymath}
    \end{small}
for all $0\leq j\leq n-1-\mathrm{L}$ and each integer $u\geq 1$. Consequently, the degree of $(f)_{0}$ is at least $n-\mathrm{L}$. In addition, by Lemma \ref{1}, the pole divisor of $f$ is bounded above by
\[p^{h-1}m^{*}_{Q}(z)\bigl(\tau^{-\mathrm{L}}(Q)+\sum^{\mathrm{L}-1}_{i=0}m\tau^{-i}(Q)\bigr),\]
which can be expressed as
\[
\deg(f)_{\infty}\leq p^{h-1}\cdot m^{*}_{Q}(z)\cdot d\cdot(1+\mathrm{L}m).
\]
Combining these bounds yields $n-\mathrm{L}\leq p^{h-1}\cdot m^{*}_{Q}(z)\cdot d\cdot(1+\mathrm{L}m)$, and thus
\begin{displaymath}
    \mathrm{NL}_{m}(\mathbf{s}_{z})\geq\frac{n-p^{h-1}m^{*}_{Q}(z)d}{1+mp^{h-1}m^{*}_{Q}(z)d}.
\end{displaymath}
Then we have the desired result.
\end{proof}
\label{7}
\end{theorem}

\section{Example construction over the rational function fields}
	\label{sec:4}
In this section, we apply the previous results to sequences defined over rational function fields. We begin by fixing the following notation:
\begin{enumerate}
    \item $\epsilon$ is a fixed primitive element of $\mathbb{F}_{q}$;
    \item $F=\mathbb{F}_{q}(x)$ denotes the rational function field;
    \item $\phi$ is the automorphism of $F/\mathbb{F}_{q}$ given by $x\mapsto \epsilon x$;
    \item $P$ is the unique zero of $x-1$.
\end{enumerate}
For each $i\in\mathbb{Z}$, let $P_{i}$ be the unique zero of
\[
\phi^{i}(x-1)=\epsilon^{i}x-1=\epsilon^{i}(x-\epsilon^{-i}),
\]
and set $n=q-1$. Then, by the cyclic structure of the finite field, the places $P_{i},P_{i+1},\ldots,P_{i+n-1}$ are mutually distinct, and we have $P_{j}=P_{j+n}$ for all $j\in\mathbb{Z}$. Let $\mathbb{P}_{d,F}$ denote the set of all places of degree $d\ge 2$. In the rational function field, there is a bijection between $\mathbb{P}_{d,F}$ and the set of monic irreducible polynomials of degree $d$ in $\mathbb{F}_{q}[x]$. The number of such places (or equivalently, irreducible polynomials) is
\begin{displaymath}
    I_{q}(d)=\frac{1}{d}\sum_{b\mid d}\mu\!\left(\frac{d}{b}\right)q^{b},
\end{displaymath}
where $\mu(\cdot)$ denotes the M$\ddot{o}$bius function.

The next lemma will be needed later; a proof is given in \cite{5}.
\begin{lemma}
    Let $Q\in\mathbb{P}_{d,F}$. If $d\geq 2$ and $\gcd(d,q-1)=1$, then for any fixed $i\in\mathbb{Z}$, the places $\phi^{i}(Q),\phi^{i+1}(Q),\ldots,\phi^{i+n-1}(Q)$ are mutually distinct.
\end{lemma}

By this lemma, we conclude that there is an integer $d \geq 2$ for which the set $\mathbb{P}_{d,F}$ is divided into $r = \frac{I_q(d)}{n}$ equivalence classes, each containing $n$ distinct places of degree $d$. Consequently, we obtain $r$ distinct places of degree $d$, which we label $Q_{1}, Q_{2}, \ldots, Q_{r}$.

For any $1\leq i\neq j\leq r$, it follows that
\begin{displaymath}
    Q_{j}\not\in\{\phi^{s}(Q_{i})\mid s\in\mathbb{Z}\}=\{Q_{i},\phi(Q_{i}),\ldots,\phi^{n-1}(Q)\}.
\end{displaymath}
Now, for each $1\leq i\leq r$, take a rational function $z\in\mathcal{L}(Q_{i})\setminus\mathbb{F}_{q}$, and consider the family of $p$-ary sequences
\begin{displaymath}
    \mathbf{S}_{d}=\{\mathbf{s}_{z}\mid z\in\mathcal{L}(Q_{i})\setminus\mathbb{F}_{q},\ 1\leq i\leq I_{q}(d)/(q-1)\},
\end{displaymath}
where
\begin{displaymath}
    \mathbf{s}_{z}=\{\mathrm{Tr}(z(P_{j}))\}_{j=0}^{\infty}=\{\mathrm{Tr}(z(\epsilon^{j}))\}_{j=0}^{\infty}.
\end{displaymath}
Therefore, the family $\mathcal{S}_{d}$ has a size of $r(q^{d}-q)$.
\begin{theorem}
With the above notation, assume that $2\leq d\leq (q-2\sqrt{q}-1)/4\sqrt{q}$ and $\gcd(d,q+1+t)=1$. Let $\mathbf{S}_{d}$ be the family of $p$-ary sequences constructed as above. Then $\mathbf{S}_{d}$ has cardinality $I_{q}(d)/(q-1)$, and each sequence in $\mathbf{S}_{d}$ has period $n=q-1$. Moreover
\begin{displaymath}
\begin{split}
    &\mathrm{L}(\mathbf{s}_{z_{i}})\geq\frac{q-1-2(d-1)\sqrt{q}}{2d\sqrt{q}};\\
    & \left|C_{\mathbf{s}_{z_{i}},\mathbf{s}_{z_{j}}}(\tau)\right|\leq2(2d-1)\sqrt{q};
    \end{split}
\end{displaymath}
for all $z_{i}\neq z_{j}\in\mathcal{L}(Q_{i})\setminus\mathbb{F}_{q}$ for $1\leq i\leq r=I_{q}(d)/(q-1)$.
    \begin{proof}
       From the condition $d\leq (q-2\sqrt{q}-1)/4\sqrt{q}$, we obtain
\[
\left(2g-2+2(m^{*}_{Q}(z)+1)d\right)\sqrt{q}
=2(-1+2d)\sqrt{q}<n.
\]
Hence, by Proposition \ref{3} and the assumption $\gcd(d,q+1+t)=1$, every sequence in $\mathbf{S}_{d}$ has period $n=q+1+t$. Choose $\mathbf{s}_{z}\in\mathcal{L}(Q_{i})\setminus\mathbb{F}_{q}$ for some $1\leq i\leq r=I_{q}(d)/(q-1)$. Then $Q_{i}$ is the unique pole of $z$, and the points $Q_{i},\phi(Q_{i}),\ldots,\phi^{n-1}(Q_{i})$ are all distinct. Applying Theorem \ref{10}, we obtain
        \begin{displaymath}
\mathrm{L}(\mathbf{s}_{z})\geq\frac{q-1-2(d-1)\sqrt{q}}{2d\sqrt{q}}
            \end{displaymath}
      for all $z\in\mathcal{L}(Q_{i})\setminus\mathbb{F}_{q}$ and $1\leq i\leq r$.
Next, we examine the correlation. Let $\mathbf{s}_{z_{i}}$ and $\mathbf{s}_{z_{j}}$ be two sequences in $\mathbf{S}_{d}$ (with the possibility that $j=i$). For any $\tau\in\mathbb{Z}$, define
\[
    f = z_{i} + \phi^{-\tau}(z_{j}),
\]
where $Q_{i}$ is the unique pole of $z_{i}$ and $\phi^{-\tau}(Q_{j})$ is the unique pole of $\phi^{-\tau}(z_{j})$.

Case 1: $i\neq j$. In this case, $Q_{i}\neq \sigma^{-\tau}(Q_{j})$ because $Q_{j}\notin\{\phi^{s}(Q_{i}) \mid s\in\mathbb{Z}\}$. Hence $Q_{i}$ cannot be a pole of $\phi^{-\tau}(z_{j})$, and we obtain $v_{Q_{i}}(f)=-1$. Thus, $f$ is non-degenerate.

Case 2: $i=j$ and $0<\tau<n$. Then $Q_{i}\neq \phi(Q_{i})$, so the same reasoning as in Case 1 applies.

Consequently, the correlation satisfies
\[
\begin{split}
    \left|C_{\mathbf{s}_{z_{i}},\mathbf{s}_{z_{j}}}(\tau)\right|
    &\leq \left(2g-2+\sum_{k=i,j}\bigl(m^{*}_{Q_{k}}(z_{k})+1\bigr)d\right)\sqrt{q}\\
    &= 2(2d-1)\sqrt{q}.
\end{split}
\]
This yields the desired result.
            \end{proof}
            \label{99999}
\end{theorem}

\begin{theorem}
Notations as in Theorem \ref{99999}, then we have
     \begin{displaymath}
        |2\mathrm{L}_{t}(\mathbf{s}_{z})-t|\leq \frac{t-\left(-2+(t+2)d\right)\sqrt{q}}{d\sqrt{q}}\ \ \  \text{for any}\ 1\leq t\leq q-1.
    \end{displaymath}
    which implies that the sequence $\mathbf{s}_{z}$ is perfect if $d\geq \frac{t+2\sqrt{q}}{\sqrt{q}(t+3)}$ and therefore for any $d\geq 2$, the family of sequences $S=\{\mathbf{s}_{z}| z\in\mathcal{L}(Q)\setminus\{0\}\}$ is perfect.
    \label{11}
    \end{theorem}

\begin{remark}
Now, we consider the finite field $\mathbb{F}_{q}$ with $q=2^{n}$ for some positive integer $n$. Then we have a family of binary sequences. Denote by $L_{1}$ and $C_{1}$ the bounds of the linear complexity and correlation in Theorem \ref{11} and denote by the bound for linear complexity and correlation of the binary sequences in \cite{3} by
   \begin{displaymath}
\begin{split}
    &\mathrm{L}(\mathbf{s}_{z_{i}})\geq\frac{q-3-2(d-1)\sqrt{q}}{2d\sqrt{q}}:=L'_{1};\\
    & \left|C_{\mathbf{s}_{z_{i}},\mathbf{s}_{z_{j}}}(\tau)\right|\leq2(2d-1)\sqrt{q}+6:=C'_{1};
    \end{split}
\end{displaymath}
for all $z_{i}\neq z_{j}\in\mathcal{L}(Q_{i})\setminus\mathbb{F}_{q}$ for $1\leq i\leq r=I_{q}(d)/(q-1)$. Then we have
\begin{displaymath}
\begin{split}
    L_{1}-L'_{1}&=\frac{q-1-2(d-1)\sqrt{q}}{2d\sqrt{q}}-\frac{q-3-2(d-1)\sqrt{q}}{2d\sqrt{q}}\\
    &=\frac{2}{2d\sqrt{q}}>0\\
    C_{1}-C'_{1}&=2(2d-1)\sqrt{q}-[2(2d-1)\sqrt{q}+6]=-6<0
    \end{split}
\end{displaymath}
Then we have $L_{1}>L'_{1}$ and $C_{1}<C'_{1}$.
\end{remark}
    
    \section{Example construction over the cyclic elliptic function fields}
    \label{sec:5}
In this section, we apply the previous results to sequences constructed over cyclic elliptic function fields.

Let $E/\mathbb{F}_{q}$ be a cyclic elliptic function field with order $N=q+1+t$ for some integer $t$ satisfying $|t|\leq 2\sqrt{q}$, and let $O$ be defined as in the Preliminaries. For any integer $d\geq 2$ with $\gcd(d,q+1+t)=1$, denote by $\mathbb{P}_{E,d}$ the set of all places of degree $d\geq 2$ of $E/\mathbb{F}_{q}$. The set $\mathbb{P}_{E,d}$ decomposes into $r=B_{d}/N$ orbits under the action of the translation group $T_{E}$. Hence, we obtain $r$ distinct places of degree $d$,
\[
    Q_{1},Q_{2},\dots,Q_{r}.
\]
It is immediate that for $1\leq i\neq j\leq r$ and for any rational place $P$,
\[
    Q_{j}\notin\{\sigma^{s}(Q_{i})\mid s\in\mathbb{Z}\}=\{Q_{i},\sigma_{P}(Q_{i}),\dots,\sigma_{P}^{N-1}(Q_{i})\}.
\]

Now take a rational function $z\in E\setminus\{0\}$ such that $v_{P_{i}}(z)\geq 0$, where $Q$ is the unique pole of $z$, and assume $\gcd(v_{Q}(z),p)=1$. As before, we can define the sequence $\mathbf{s}_{z}$ and consider the sequence family
\[
    \tilde{\mathbf{S}}_{d}=\{\mathbf{s}_{z}\mid z\in E\setminus\mathbb{F}_{q},\ \gcd(v_{Q}(z),p)=1,\ 1\leq i\leq r=B_{d}/(q+1+t)\},
\]
where
\[
    \mathbf{s}_{z}=\{\operatorname{Tr}(z(P_{j}))\}_{j=0}^{\infty}=\{\operatorname{Tr}(z([j]P))\}_{j=0}^{\infty}.
\]

Let $z_{i}$ and $z_{j}$ be rational functions for which $Q_{i}$ (respectively $Q_{j}$) is the unique pole of $z_{i}$ (respectively $z_{j}$), and such that $\gcd(v_{Q_{i}}(z_{i}),p)=1$ (respectively $\gcd(v_{Q_{j}}(z_{j}),p)=1$). Under these assumptions, we obtain the following theorems.
\begin{example}
 Let $t$ be an integer satisfying one condition in Lemma \ref{222} .If
    \begin{displaymath}
        N>2(m^{*}_{Q}(z)+1)d\sqrt{q}.
        \end{displaymath}
    In this case, the period of $\mathbf{s}_{z}$ is $N=q+1+t$.
    
    In particular, if $m^{*}_{Q}(z)=1$, this coincides with the result of Hu \textit{et al}. in \cite{22}. That is, if
    \[
    N>4d\sqrt{q},
    \]
    then the period of $\mathbf{s}_{z}$ is $N=q+1+t$.
    \label{4}
    \end{example}

\begin{theorem}
Let $t$ be an integer satisfying one condition in Lemma \ref{222} and $2\leq d\leq (q+1+t)/2(m^{*}_{Q_{i}}(z)+1)\sqrt{q}-1$, $\gcd(d,q+1+t)=1$, and $\gcd(v_{Q_{i}}(z),p)=1$ for $1\leq i\leq r=B_{d}/(q+1+t)$. Let $\tilde{\mathbf{S}}_{d}$ be the family of $p-$ary sequences as constructed above. Then for any $\mathbf{s}_{z_i}, \mathbf{s}_{z_j}\in\tilde{\mathbf{S}}_{d}$, we have
$$
\mathrm{L}(\mathbf{s}_{z_{i}})\geq\frac{q+1+t-(m^{*}_{Q_{i}}(z_{i})+1)d\sqrt{q}}{(m^{*}_{Q_{i}}(z_{i})+1)d\sqrt{q}}
$$
and
$$
\left|C_{\mathbf{s}_{z_{i}},\mathbf{s}_{z_{j}}}(\tau)\right|\leq\sum_{k=i,j}(m^{*}_{Q_{k}}(z_{k})+1)d_{k}\sqrt{q}.
$$
    \begin{proof}
       From the assumption $2\leq d\leq (q+1+t)/\bigl(2(m^{*}_{Q_{i}}(z)+1)\sqrt{q}-1\bigr)$, we obtain the inequality $2(m^{*}_{Q}(z)+1)d\sqrt{q}<q+1+t$. Hence, by Proposition \ref{3}, every sequence in $\mathbf{S}_{d}$ has period $n=q+1+t$. Now take $\mathbf{s}_{z_{i}}\in\tilde{\mathbf{S}}_{d}$ for some $1\leq i\leq r=B_{d}/(q+1+t)$. According to the condition on $Q_{i}$, this place is the sole pole of $z_{i}$, and the points $Q_{i},\sigma_{P}(Q_{i}),\ldots,\sigma_{P}^{N-1}(Q_{i})$ are all distinct from one another. Applying Theorem \ref{10}, we conclude that
        \begin{displaymath}
\mathrm{L}(\mathbf{s}_{z_{i}})\geq\frac{q+1+t-(m^{*}_{Q_{i}}(z_{i})+1)d\sqrt{q}}{(m^{*}_{Q_{i}}(z_{i})+1)d\sqrt{q}}
            \end{displaymath}
        for any  $z_{i}\in E$ satisfying $Q_{i}$, there is a unique pole of $z_{i}$ and $\gcd(v_{Q_{i}}(z_{i}),p)=1$.
        
           Now consider the correlation. Let $\mathbf{s}_{z_{i}}$ and $\mathbf{s}_{z_{j}}$ be two sequences of $\mathbf{S}_{d}$ (it allows for $i=j$). For $\tau\in\mathbb{Z}$, consider the function 
        \begin{displaymath}
            f=z_{i}+\sigma_{P}^{-\tau}(z_{j})
        \end{displaymath}
        $Q_{i}$ is the unique pole of $z_{i}$, and $\sigma_{P}^{-\tau}(Q_{j})$ is the unique pole of $\sigma_{P}^{-\tau}(z_{j})$.

Case 1: $i\neq j$. then $Q_{i}\neq\sigma^{-\tau}(Q_{j})$ since $Q_{j}\notin\{\sigma^{s}(Q_{i})| s\in\mathbb{Z}\}$. Therefore, $Q_{i}$ is not the pole of $\sigma^{-\tau}(z_{j})$ and $v_{Q_{i}}(f)=-e_{i}$ with $\gcd(e_{i},p)=1$; $f$ is non-degenerate.

Case 2: $i= j$. For any $0<\tau<q+1+t$, we have $Q_{i}\neq \sigma^{-\tau}(Q_{i})$ and the same discussion as in case~1.

Then the correlation can be expressed as
\begin{displaymath}
\begin{split}
     \left|C_{\mathbf{s}_{z_{i}},\mathbf{s}_{z_{j}}}(\tau)\right|&\leq   
 \left(2g(E)-2+\sum_{k=i,j}(m^{*}_{Q_{k}}(z_{k})+1)d_{k}\right)\sqrt{q}\\
 &=\sum_{k=i,j}(m^{*}_{Q_{k}}(z_{k})+1)d_{k}\sqrt{q}.
 \end{split}
 \end{displaymath}
            \end{proof}
            \label{999999}
\end{theorem}

\begin{theorem}
    Notations as in Theorem \ref{999999}, we have
    \begin{displaymath}
        |2\mathrm{L}_{s}(\mathbf{s}_{z})-s|\leq \frac{s-  (s+2)d\sqrt{q}}{d\sqrt{q}}\ \ \ \text{for any}\ 1\leq s\leq q+1+t        \end{displaymath}
    which implies the sequence  $\mathbf{s}_{z}$ is perfect if $d\geq \frac{s}{\sqrt{q}(s+3)}$ and it also follows that for any $d\geq 2$, the family of sequences $S=\{\mathbf{s}_{z}| z\in\mathcal{L}(Q)\setminus\{0\}\}$ is perfect.
    \end{theorem}

\begin{remark}
    Also consider the finite field $\mathbb{F}_{q}$ with $q=2^{n}$ for some positive integer $n$. Consider $z_{i}\in\mathcal{L}(Q_{i})$ for $1\leq i\leq B_{d}/(q+1+t)$, then we denote by $L_{2}$ and $C_{2}$ the bounds of the linear complexity and correlation derived from Theorem \ref{999999}, \textit{i.e.}
    \begin{displaymath}
        \begin{split}
            L_{2}&=\frac{q+1+t-2d\sqrt{q}}{2d\sqrt{q}}\\
            C_{2}&=4d\sqrt{q}
        \end{split}
    \end{displaymath}
    Denote the bound for linear complexity and correlation of the binary sequences in \cite{3} by
    \begin{displaymath}
        \begin{split}
            \mathrm{L}(\mathbf{s}_{z_{i}})&\geq \frac{q+1+2t-2(d+1)\sqrt{q}}{2d\sqrt{q}}:=L'_{2}\\
            C(\mathbf{s}_{z_{i}})&\leq 2(2d+1)\sqrt{q}+|t|:=C'_{2}
        \end{split}
    \end{displaymath}
Then we have 
\begin{displaymath}
\begin{split}
    L_{2}-L'_{2}&=\frac{q+1+t-2d\sqrt{q}}{2d\sqrt{q}}-\frac{q+1+2t-2(d+1)\sqrt{q}}{2d\sqrt{q}}\\
    &=\frac{-t+2\sqrt{q}}{2d\sqrt{q}}\\
    C_{2}-C'_{2}&=4d\sqrt{q}-2(2d+1)\sqrt{q}-|t|\\
    &=-2\sqrt{q}-|t|<0
    \end{split}
    \end{displaymath}
and the equality $L_{2}=L'_{2}$ holds if and only if $n$ is even and
\begin{displaymath}
        t=2\sqrt{q}
\end{displaymath}
which means $q+1+t=(2^{n/2}+1)^{2}$ and such an elliptic curve does not exist according to Theorem 3 in \cite{8}. Then we have $L_{1}>L'_{1}$ and $C_{2}<C'_{2}$.
\end{remark}

    \section{Conclusion}
    \label{sec:7}
    In this paper, we generalize the construction of $p-$ary sequences from elliptic function fields of Hu \textit{et al.} in~\cite{22} to the constructions by using exponential sums over general algebraic function fields and results given by Deligne~\cite{32}, Weil~\cite{33}, and Bombieri~\cite{33}. Our results show that these sequences have large periods, large linear complexity, low correlation, and  high nonlinear complexity, which extend the results in~\cite{20,5} to the $p-$ary sequences from general algebraic function fields. 
	
        \end{document}